\documentclass[reprint,
groupedaddress,
nobibnotes,
 amsmath,amssymb,
 aps,
prb,
]{revtex4-1}

\usepackage{amsfonts, mathrsfs, graphicx}
\usepackage{bbm}
\usepackage{dcolumn}% Align table columns on decimal point

\usepackage{umoline}
\usepackage[usenames,svgnames]{xcolor}

\usepackage{hyperref}
\hypersetup{
     colorlinks=true,           % false: boxed links; true: colored links
     linkcolor=Crimson,          % color of internal links
     citecolor=blue,            % color of links to bibliography
     filecolor=blue,            % color of file links
     urlcolor=cyan,             % color of external links
}

\newtheorem{theorem}{Theorem}

\newtheorem{lemma}[theorem]{Lemma}
\newtheorem{lemma*}{Lemma}

\newtheorem{claim}[theorem]{Claim}

\newtheorem{corollary}[theorem]{Corollary}
\newtheorem{definition}[theorem]{Definition}

\newcommand{\qedsymb}{\hfill{\rule{2mm}{2mm}}}  
\newenvironment{proof}[1][]{\begin{trivlist}  
  \item[\hspace{\labelsep}{\it\noindent Proof#1:\/}]}
  {\qedsymb\end{trivlist}}

\newcommand{\Fig}[1]{Fig.~\ref{#1}}

\newcommand{\Lem}[1]{Lemma~\ref{#1}}
\newcommand{\Cor}[1]{Corollary~\ref{#1}}
\newcommand{\Sec}[1]{Sec.~\ref{#1}}
\newcommand{\Ref}[1]{Ref.~[\onlinecite{#1}]}

\newcommand{\Thm}[1]{Theorem~\ref{#1}}

\newcommand{\cc}[1]{\cite{#1}}

\newcommand{\bigO}[1]{\ensuremath{\operatorname{O}\bigl(#1\bigr)}}

\newcommand{\bOmega}[1]{\ensuremath{\operatorname{\Omega}\bigl(#1\bigr)}}

\newcommand{\ket}[1]{|#1\rangle}
\newcommand{\bra}[1]{\langle#1|}

\newcommand{\norm}[1]{\|#1\|}
\newcommand{\bnorm}[1]{\big\|#1\big\|}

\newcommand{\EqDef}{:=}

\newcommand{\DL}{\mathop{\rm DL}\nolimits}
\newcommand{\Id}{\ensuremath{\mathbbm{1}}}

\newcommand{\eps}{\epsilon}
\newcommand{\spg}{\gamma}
\newcommand{\gs}{\Omega}

\newcommand{\oP}{\bar{P}}
\newcommand{\oQ}{\bar{Q}}
\newcommand{\oH}{\bar{H}}

\newcommand{\ospg}{\bar{\gamma}}
\newcommand{\oPi}{\bar{\Pi}}
\newcommand{\Pigs}{\Pi_{gs}}
\newcommand{\Vgs}{V_{gs}}

\begin{document}

\title{A simple proof of the detectability lemma and 
spectral gap amplification}

\author{Anurag Anshu}
\email{a0109169@u.nus.edu}
\affiliation{Centre for Quantum Technologies,
  National University of Singapore, Singapore}
\author{Itai Arad}
\email{arad.itai@fastmail.com}
\affiliation{Centre for Quantum Technologies, National University 
  of Singapore, Singapore}
\author{Thomas Vidick}
\email{vidick@cms.caltech.edu}
\affiliation{Department of Computing and Mathematical Sciences,
  California Institute of Technology, Pasadena, USA}
\date{\today} 

\begin{abstract}
  The detectability lemma is a useful tool for probing the structure
  of gapped ground states of frustration-free Hamiltonians of
  lattice spin models. The lemma provides an estimate on the error
  incurred by approximating the ground space projector with a
  product of local projectors. We provide a new, simpler proof for
  the detectability lemma which applies to an arbitrary ordering of
  the local projectors, and show that it is tight up to a constant
  factor. As an application, we show how the lemma can be combined
  with a strong converse by Gao to obtain local spectral gap
  amplification: We show that by coarse graining a local
  frustration-free Hamiltonian with a spectral gap $\spg>0$ to a
  length scale $\bigO{\spg^{-1/2}}$, one gets a Hamiltonian with an
  $\bOmega{1}$ spectral gap.
\end{abstract}

\maketitle

%%%%%%%%%%%%%%%%%%%%%%%%%%%%%%%%%%%%%%%%%%%%%%%%%%%%%%%%%%%%%%%%%
\section{Introduction}
\label{sec:intro}

In recent years our understanding of quantum many-body systems, and
in particular the properties of their ground states, has shown
considerable progress. Much of this understanding can be attributed
to the development of new technical tools for analyzing general
many-body quantum systems. A particularly powerful set of
techniques, pioneered by Hastings\cc{ref:Hastings2004-EXP}, uses 
Lieb-Robinson bounds\cc{ref:LRB72,ref:Nachtergaele2006-LR} together
with appropriate filtering functions to construct local
approximations to the action of the ground state projector. These
techniques were successfully leveraged to rigorously establish many
interesting properties of ground states such as exponential decay of
correlations in gapped models\cc{ref:Hastings2004-EXP,
ref:HK2006-EXP, ref:Nachtergaele2006-LR}, an area law for
one-dimensional (1D) gapped systems\cc{ref:Hastings2007-AL},
efficient classical simulation of adiabatic evolution of 1D gapped
systems\cc{Osborne2007-adiabatic1D, Hastings2009-adiabatic1D},
stability of topological order\cc{ref:BHM2010-TQO-stability,
ref:BH2011-TQO-stability}, classification of quantum
phases\cc{ref:CGW2010-classification}, and many more (see,
e.g.,~\Ref{ref:Hastings2010-locality} and references therein).

More recently, originating in an attempt to tackle some aspects of
the quantum PCP conjecture\cc{ref:AharonovAV13qpcp}, a new tool has
been introduced for the analysis of many-body local Hamiltonians,
known as the \emph{detectability lemma
(DL)}\cc{ref:AALV2009-gap-amp}. The DL has proven particularly
useful for studying the ground states of gapped,
\emph{frustration-free} spin systems on a
lattice\cc{ref:AharoAVZ2011-DL}. Examples of such systems include
the Affleck-Kennedy-Lieb-Tasaki (AKLT) model \cc{ref:AKLT87}, the
spin $1/2$ ferromagnetic XXZ chain\cc{ref:NK1997-XXZ} and Kitaev's
toric code \cc{ref:KitaevToricCode97-1, ref:KitaevToricCode97-2}.

Given a local Hamiltonian $H$ that is frustration free, the
detectability lemma operator $\DL(H)$ is defined as a product of the
local ground space projectors associated to each term in the
Hamiltonian, organized in layers (see \Fig{fig:layers} and
\Sec{sec:DL} for a precise statement). The DL operator leaves the
ground space of $H$ invariant while shrinking all excited states by
a factor of at least $1-\Delta$ for some $0<\Delta<1$.  The
detectability lemma establishes a lower bound on $\Delta$, thereby
placing an upper bound on the shrinking of any state orthogonal to
the ground space. Essentially, the lemma shows that $\Delta$ is at
least a constant times the spectral gap of $H$. 

Since the DL operator preserves the ground space and shrinks any
state orthogonal to it, it can be viewed as an approximation to the
ground state projector, with an error of $1-\Delta$. This allows one
to approximate the highly complex and possibly non-local ground
space projector of the full system by the simpler operator $\DL(H)$
(or a power of it). It provides a considerably simpler alternative
to more general constructions based on Lieb-Robinson bound and the
use of filtering functions (admittedly those constructions also
apply to frustrated systems). Many results that were proved for
general systems using these techniques, such as the 1D area law and
the exponential decay of correlations, can be proved in simpler way
for the case of frustration-free systems using the
DL~\cc{ref:AharoAVZ2011-DL}. In addition, the DL has found further
applications such as the analysis of
T-designs\cc{ref:BHH2012-Tdesign} and Gibbs
samplers\cc{ref:KB2014-Gibbs}, and an improvement to the original 1D
area law for frustration-free systems\cc{ref:ALV2012-AL}.

The original proof of the DL from \Ref{ref:AALV2009-gap-amp} used
the so-called XY decomposition and was limited to local Hamiltonians
in which the local terms are taken from a constant set.
Subsequently, a much simpler proof, which does not rely on the XY
decomposition and is free of the limitations of the first proof, was
introduced in \Ref{ref:AharoAVZ2011-DL}. In this paper we introduce
yet another proof of the DL, which is simpler than the proof of
\Ref{ref:AharoAVZ2011-DL}, provides a tighter bound on $1-\Delta$,
and is more general as it holds for an \emph{arbitrary} ordering of
the local projectors. This tighter form of the DL has already been
used in~\cite{ref:GH2015-DLAmp} to derive a quadratically improved
upper bound on the correlation length of gapped ground states of
frustration-free systems.

Recent work of Gao\cc{ref:gao2015quantum} on a quantum union bound
establishes a \emph{converse} to the DL that provides a lower bound
on the spectral gap of a frustration-free Hamiltonian $H$ as a
function of the spectral gap of $\DL(H)$.\footnote{A previous arXiv
version of this paper contained a proof for a slightly weaker
statement than Gao's.} Equivalently, Gao's result places an upper
bound on the parameter $\Delta$, or a lower bound on the shrinking
of excited states by $\DL(H)$ (see Lemma~\ref{lem:LD} for a precise
statement). Together with the detectability lemma, the two results
establish a form of duality between $H$ and $\DL(H)$, showing that
their spectral gaps are always within a constant factor from each
other. This converse to the DL has already been used for the purpose
of proving lower bounds on the spectral gap of frustration-free
Hamiltonians in forthcoming work on 1D area laws and efficient
algorithms\cc{ref:BigAL}. 

As an application, in the second part of this paper we show how a
combination of the DL and its converse can be used to prove that the
spectral gap of a local frustration-free Hamiltonian can be
amplified from $\spg>0$ to a constant by coarse-graining the
Hamiltonian to a length scale $\bigO{\spg^{-1/2}}$. A direct
application of both lemmas provides the result for a length scale
$\bigO{\spg^{-1}}$; we quadratically improve the dependence on
$\spg$ by employing a Chebyshev polynomial in a way analogous to
recent work of Gosset and Huang~\cite{ref:GH2015-DLAmp}.  

\paragraph*{Organization.} In \Sec{sec:DL} we state and prove the DL.
In \Sec{sec:gap} we give our application to spectral gap
amplification.

%%%%%%%%%%%%%%%%%%%%%%%%%%%%%%%%%%%%%%%%%%%%%%%%%%%%%%%%%%%%%%%%%
\section{The DL operator and frustration-free spin systems on a
lattice} \label{sec:setup} 

Throughout we use the ``big O'' notation, where $\bigO{f(x)}$
indicates any function $g$ such that there is a constant $C>0$,
$|g(x)|\leq C\cdot f(x)$ for all $x$ in the domain of $f$.
Similarly, $\bOmega{f(x)}$ denotes any function $g$ such that there
exists a constant $c>0$ such that $g(x) \geq c\cdot f(x)$ for all
$x$ in the domain of $f$.

We concentrate on \emph{frustration-free} spin systems on regular
lattices. Formally, we consider $n$ quantum spins with local
dimension $d$ that are positioned on the vertices of a regular
$D$-dimensional lattice with an underlying Hilbert space
$\mathcal{H} = (\mathbb{C}^d)^{\otimes n}$. On this lattice we
consider a $k$-local Hamiltonian system $H=\sum_i h_i$ where each
$h_i$ acts on at most $k$ neighboring spins of the lattice. It is
easy to see that in this setting every local term does not commute
with at most $g$ other local terms, where $g$ is a constant.
Moreover, the set of local terms can always be partitioned into $L$
subsets $T_1, T_2, \ldots, T_L$, called \emph{layers}, such that
each layer consists of non-overlapping local terms, which are
therefore pairwise commuting. Clearly, both $g$ and $L$ can be upper
bounded as functions of $k$ and $D$ [trivial bounds are $g\leq
k(2D)^{k-1}$ and $L\leq (2D)^{2k}$]; for clarity, here we treat them
as independent parameters. A canonical example is a spin chain over
$n$ spins with nearest-neighbor interactions
$H=\sum_{i=1}^{n-1}h_i$, where $h_i$ acts on spins $\{i,i+1\}$. Each
$h_i$ is non-commuting with at most $g=2$ neighbors, and the system
can be partitioned into $L=2$ layers, the \emph{odd} layer $T_{\text{odd}}
=\{ h_1, h_3, h_5, \ldots\}$ and the complementary \emph{even} layer
$T_{\text{even}}$. This decomposition is illustrated in \Fig{fig:layers}. 

\begin{figure}
  \centering
  \includegraphics[scale=0.9]{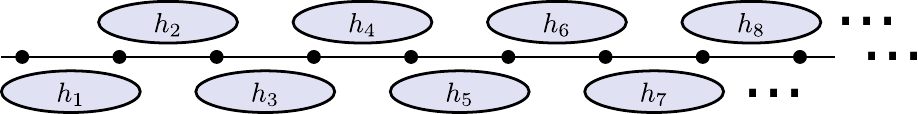} 
  \caption{Decomposing the local terms of a 1D Hamiltonian $H=\sum_i
  h_i$ with 2-local, nearest-neighbor interactions into two layers
  --- an even layer and an odd layer. \label{fig:layers} }
\end{figure}

By adding constant multiples of the identity to each $h_i$ we may
assume without loss of generality that their smallest eigenvalue is
$0$. Moreover, assuming that the norms of the $h_i$ are uniformly
bounded by a constant, we may scale the system and switch to
dimensionless units in which $\norm{h_i}\le 1$ and therefore $0\le
h_i \le \Id$.  We label the energy levels of $H$ by $\eps_0 < \eps_1
< \eps_2 \cdots$, where each level may correspond to more than one
eigenstate of $H$. The ground space of $H$ is denoted by $\Vgs$ and
the projector onto it by $\Pigs$. We let $\spg\EqDef\eps_1-\eps_0>0$
denote the \emph{spectral gap} of the system. 

We say that the system is \emph{frustration free} when every ground
state $\ket{\gs}\in\Vgs$ minimizes the energy of each local term
$h_i$ {separately}, i.e., $\bra{\gs}h_i\ket{\gs}=0$.  Notice that in
such case it necessarily holds that $h_i\ket{\gs}=0$ and hence every
ground state is a {common eigenstate} of all $h_i$. This property
strongly constrains the structure of frustration-free ground states
and makes their analysis much simpler in comparison with the general
frustrated case. 

When studying frustration-free ground states it is often convenient
to introduce an auxiliary Hamiltonian in which every $h_i$ is
replaced by a projector $Q_i$ whose null space coincides with the
null space of $h_i$. The auxiliary Hamiltonian $\hat{H}\EqDef \sum_i
Q_i$ and the original Hamiltonian $H=\sum_i h_i$ thus share the same
ground space. Moreover, since $0\le h_i \le \Id$, $0\le h_i \le
Q_i$, and $\hat{H}\ge H$. It follows that if $H$ is gapped, then so
is $\hat{H}$, with $\spg(\hat{H})\ge \spg(H)$. Note that in case the
original Hamiltonian $H = \sum_i a_i Q_i$ with $Q_i$ projectors and
$0\leq a_i\leq 1$, the effect of this transformation is simply to
set $\hat{H} = \sum_i Q_i$, a Hamiltonian with the same ground space
and a gap at least as large as that of $H$.  From here onwards in
order to keep the notation light we shall denote $\hat{H}$ by $H$,
or simply assume that $H$ itself is given as a sum of projectors,
$H=\sum_i Q_i$.

A useful approach for understanding the locality properties of the
ground space of $H$ consists in approximating its ground state
projector $\Pigs$ by an operator that possesses a more local
structure, and is therefore easier to work with. Such operators are
referred to as Approximate Ground State Projectors (AGSPs), and
various constructions have been used to establish properties of
gapped ground states such as exponential decay of
correlations\cc{ref:AharoAVZ2011-DL,ref:GH2015-DLAmp}, area
laws\cc{ref:ALV2012-AL, ref:AKLV2013-AL}, and local
reversibility\cc{ref:Kuwahara2015-LR}.  Frustration-free systems can
be given a very natural construction of AGSP, called the
\emph{detectability lemma operator} $\DL(H)$. To introduce this
operator, define the \emph{layer projector} $\Pi_\ell\EqDef
\prod_{i\in T_\ell} (\Id-Q_i)$ for every layer $\ell$. As $\Pi_\ell$
is a product of commuting projectors, it is by itself a projector
--- the projector onto the ground space of the $\ell$-th layer. Then
$\DL(H)$ is defined as follows. 

\begin{definition}[The detectability lemma operator]
\label{def:dl}
  Given a decomposition of the terms of a local Hamiltonian
  $H=\sum_i Q_i$ in $L$ layers $T_1,\ldots,T_L$ the detectability
  lemma operator of $H$ is defined as
  \begin{align}
  \label{def:DL-operator}
    \DL(H) \EqDef \Pi_L\cdots \Pi_1 
      =  \prod_{\ell=1}^L \prod_{i\in T_\ell} (\Id-Q_i).
  \end{align}
\end{definition}
It is easy to see that $\DL(H)$ is indeed an AGSP: by the
frustration-free assumption each $\Id-Q_i$ preserves the ground
space, hence $\DL(H)\Vgs = \Vgs$. Moreover, $\norm{\DL(H)}\le 1$,
since its a product of projectors, and $\norm{\DL(H)\ket{\psi}}=1$
if and only if $\ket{\psi}\in \Vgs$. Therefore, there exists some
$0<\Delta<1$ such that for every state $\ket{\psi^\perp}$ that is
perpendicular to the ground space, $\norm{\DL(H)\ket{\psi^\perp}}\le
1-\Delta$. It follows that $\norm{\Pigs-\DL(H)}\le 1-\Delta$.
Therefore the DL operator is an AGSP, whose quality is determined by
the parameter $\Delta$. Moreover, using again the fact that the
system is frustration-free, one can amplify the quality of
approximation by taking power of the DL operator:
$\norm{\Pigs-\DL^q(H)}\le (1-\Delta)^q$ for any $q\ge 1$. 

As an operator, $\DL^q(H)$ is an alternating product of layer
projectors. Pictorially, it can be visualized as a stack of layers,
much as a brick wall (see, e.g.,~\Fig{fig:coarse-grain}). One can
verify that the collection of projectors $\Id-Q_i$ appearing in
$\DL^q(H)$ that do not commute with a given local operator $B$ forms
a ``light cone'' centered at $B$. This observation is crucial for
understanding the effect of $B$ on the ground space, and is arguably
the most important way in which locality of the DL operator can be
leveraged.

We are left with the task of estimating the parameter $\Delta$. The
detectability lemma, introduced in the next section, provides a
lower bound on $\Delta$ (an upper bound on $1-\Delta$). The converse
to the lemma, Lemma~\ref{lem:LD}, provides an upper bound on
$\Delta$. Crucially, even though both bounds depend on $\spg$,
 and the bound from the DL also depends on $g$, both bounds
are independent of the system size.

%%%%%%%%%%%%%%%%%%%%%%%%%%%%%%%%%%%%%%%%%%%%%%%%%%%%%%%%%%%%%%%%%
\section{A simple proof of the detectability lemma}
\label{sec:DL}

The variant of the DL we are about to prove is more general that the
one from~\Ref{ref:AharoAVZ2011-DL} in that the projectors $Q_i$ are
not assumed to be local, nor placed on a fixed lattice; the order of
their product in $\DL(H)$ can be {arbitrary}. Luckily, the proof
also turns out to be simpler than the original proof.

\begin{lemma}[The detectability lemma (DL)]
\label{lem:detect} 

  Let $\{Q_1, \ldots, Q_m\}$ be a set of projectors and $H =
  \sum_{i=1}^m Q_i$.  Assume that each $Q_i$ commutes with all but
  $g$ others. Given a state $\ket{\psi}$, define $\ket{\phi}\EqDef
  \prod_{i=1}^m (\Id-Q_i)\ket{\psi}$, where the product is taken in
  any order, and let $\eps_\phi\EqDef
  \frac{1}{\norm{\phi}^2}\bra{\phi}H\ket{\phi}$ be its energy. Then 
  \begin{align}
    \bnorm{ \prod_{i=1}^m (\Id-Q_i)\ket{\psi}}^2 
      \le   \frac{1}{\eps_\phi/g^2 + 1} .
  \end{align}
\end{lemma}

By choosing the order of the projectors to coincide with that in
$\DL(H)$ (for any decomposition into layers), and observing that for
every state $\ket{\psi^\perp}$ orthogonal to the ground space it
holds that $\bra{\psi^\perp}H\ket{\psi^\perp}\ge \spg$, we obtain
the following immediate corollary. 

\begin{corollary}
\label{cor:DL}
  For any state $\ket{\psi^\perp}$ orthogonal to the ground space of
  $H$, 
  \begin{align}
    \norm{\DL(H)\ket{\psi^\perp}}^2 \le \frac{1}{\spg/g^2 + 1} .
  \end{align}
\end{corollary}
In light of the discussion in the Introduction, we see that the DL
implies that $1-\Delta \le \frac{1}{\sqrt{\spg/g^2 + 1}}$, or,
equivalently, $\Delta\ge 1-\frac{1}{\sqrt{\spg/g^2 + 1}}\ge
\spg/(4g^2)$, where the second inequality follows from the fact that
$\spg/g^2<1$. (To see this, note that a state of energy at most
$g+1<g^2$ and orthogonal to the ground space can always be
constructed by starting from any ground state and replacing the
state of the spins associated with an arbitrary local term $h_i$
with a local state orthogonal to the ground space of $h_i$.)

We now turn to the proof of the DL; after the proof we
give a simple example showing that the dependence on $g$ in the
bound provided by the lemma is necessary.

\begin{proof}[ of Lemma~\ref{lem:detect}]
  We start by considering
  \begin{align*}
    \bra{\phi}H\ket{\phi}=\sum_{i=1}^m\bra{\phi}Q_i\ket{\phi}
      =\sum_{i=1}^m \norm{Q_i\ket{\phi}}^2.
  \end{align*}
  To bound $\norm{Q_i\ket{\phi}}$ we write it as
  $\norm{Q_i(\Id-Q_m)\cdots(\Id-Q_1)\ket{\psi}}$ and try to move
  $Q_i$ to the right until it hits $(\Id-Q_i)$ and vanishes. Let
  $N_i$ denote the subset of indices of projectors that do not
  commute with $Q_i$. Whenever $j\in N_i$, we use the triangle
  inequality to write
  \begin{align*}
    &\norm{Q_i(\Id-Q_j)\cdot(\Id-Q_{j-1})\cdots(\Id-Q_1)\ket{\psi}} \\
    &\le \norm{Q_i(\Id-Q_{j-1})\cdots(\Id-Q_1)\ket{\psi}} \\
    &\ \ +  \norm{Q_iQ_j(\Id-Q_{j-1})\cdots(\Id-Q_1)\ket{\psi}}.
  \end{align*}
  Therefore,
  \begin{align*}
    \norm{Q_i\ket{\phi}}
     \le \sum_{j\in N_i}
       \norm{Q_j(\Id-Q_{j-1})\cdots(\Id-Q_1)\ket{\psi}},
  \end{align*}
  where we also used $\norm{Q_i}\le 1$. Since $|N_i|\le g$, we get
  \begin{align*}
    \norm{Q_i\ket{\phi}}^2
     \le g\sum_{j\in N_i}
       \norm{Q_j(\Id-Q_{j-1})\cdots(\Id-Q_1)\ket{\psi}}^2.
  \end{align*}
  Summing over $i=1,\ldots,m$, each term
  $\norm{Q_j(\Id-Q_{j-1})\cdots(\Id-Q_1)\ket{\psi}}^2$ appears at
  most $g$ time because there are at most $g$ projectors $Q_i$ that
  do not commute with $Q_j$. Thus
  \begin{align*}
    \bra{\phi}H\ket{\phi} 
      &= \sum_i \bra{\phi}Q_i\ket{\phi} 
      = \sum_i\norm{Q_i\ket{\phi}}^2 \\
      &\le g^2\sum_{j=2}^m \norm{Q_j(\Id-Q_{j-1})
        \cdots(\Id-Q_1)\ket{\psi}}^2\\
      &= g^2\big[\norm{(\Id-Q_1)\ket{\psi}}^2 \\
      &\hskip1.5cm - \norm{(\Id-Q_m)\cdots(\Id-Q_1)\ket{\psi}}^2\big]\\
      &\le g^2(1-\norm{\phi}^2) ,
  \end{align*}
  where the third line follows from a telescopic sum. Writing
  $\bra{\phi}H\ket{\phi}=\norm{\phi}^2\eps_\phi$ and re-arranging
  terms proves the lemma.
\end{proof}

We end this section with a simple example showing that the
dependence on $g$ in the bound of the DL is necessary. The idea is
to consider $g$ projection operators in two dimensions, each making
a small angle $\approx \eps$ with the next one. Sequentially
applying these projections will reduce the squared norm of a certain
state by $\approx g\eps^2$, but the final state will be sufficiently
far from most of the projection operators for its energy to be
$\Omega(g^3\eps^2)$. 

We proceed with the construction. Let $\eps>0$ and $g$ a positive
integer. Let $\ket{\psi} = \ket{0} \in\mathbb{C}^2$ and for
$i\in\{1,\ldots,g\}$ let $Q_i = \ket{\varphi_i}\bra{\varphi_i}$,
where we defined $\ket{\varphi_i} = \sin \eps_i \ket{0}- \cos \eps_i
\ket{1}$ and $\eps_i = i\eps$. Let $\ket{\varphi_i^\perp} = \cos
\eps_i \ket{0}+ \sin \eps_i \ket{1}$. Applying the sequence of projections 
$(\Id-Q_1)\to (\Id-Q_2)\to \ldots\to (\Id-Q_g)$ to $\ket{\psi}$, we
obtain (up to normalization) the states $\ket{\psi_1^\perp} \to
\ket{\psi_2^\perp} \to \ldots \to \ket{\psi_g^\perp}$. To estimate
the norm of the final state, note that
\begin{align*}
  \bra{\varphi_i^\perp}\varphi_{i+1}^\perp\rangle 
    &= \cos \eps_i \cos \eps_{i+1} + \sin \eps_i \sin \eps_{i+1} \\
      &=\cos(\eps_{i+1}-\eps_i) = \cos(\eps) , 
\end{align*}
so $(\Id-Q_g)\cdots(\Id-Q_1)\ket{\psi} = (\cos \eps)^g
\ket{\varphi_g^\perp}$, with squared norm 
\begin{align}
\nonumber
  \bnorm{(1-Q_g)\cdots(1-Q_1) \ket{\psi}}^2 
    &= \cos^{2g} \eps \geq \Big(1-\frac{\eps^2}{2}\Big)^{2g} \\
      &\geq 1-2g\,\eps^2 
\label{eq:norm}
\end{align}
for small enough $\eps$. To estimate the
energy of $\ket{\psi_g^\perp}$, note that for any $i$,
\begin{align*}
  \bnorm{Q_i \ket{\psi_g^\perp}}^2 = \sin^2 (\eps_i-\eps_g) 
    \ge \frac{(g-i)^2\eps^2}{2}
\end{align*}
for small enough $\eps$, so that 
\begin{align}
\label{eq:energy}
  \sum_{i=1}^g\bnorm{Q_i \ket{\psi_g^\perp}}^2  
    &\geq \frac{1}{2}\big[ (g-1)^2+ \cdots+ 2^2 +1\big]\eps^2 \\
    &= \frac{(g-1)(g)(2g-1)}{12}\eps^2 . \nonumber
\end{align}
Combining~\eqref{eq:norm} and~\eqref{eq:energy}, for large $g$ and
small enough $\eps$, 
\begin{align*}
  \sum_{i=1}^g  \norm{Q_i \ket{\psi_g^\perp}}^2 
    \geq \frac{(g-1)^2}{12} \Big[1-  \bnorm{(1-Q_g)
      \cdots(1-Q_1) \ket{\psi}}^2\Big] ,
\end{align*}
matching the bound from \Lem{lem:detect} up to constant
factors. 

%%%%%%%%%%%%%%%%%%%%%%%%%%%%%%%%%%%%%%%%%%%%%%%%%%%%%%%%%%%%%%%%%%
\section{Spectral gap amplification}
\label{sec:gap}

In this section we show how a simple combination of the DL and its
converse\cc{ref:gao2015quantum} can be used to prove that the
spectral gap of a frustration-free Hamiltonian made from projectors
can be amplified from any $\spg>0$ to a constant by coarse graining
the Hamiltonian to a length scale of $\bigO{\spg^{-1/2}}$. Our proof
employs a recent ``trick'' by Gosset and Huang\cc{ref:GH2015-DLAmp}
to boost the effect of $\DL^q(H)$ by using a Chebyshev polynomial.
This reduces the length scale of the required coarse graining from
$O(\spg^{-1})$ to $O(\spg^{-1/2})$.

For the sake of clarity we present the result for a nearest-neighbor
Hamiltonian defined on a line of particles; extension to
higher-dimensional lattices is straightforward. Let $H=\sum_i Q_i$
be a nearest-neighbor frustration-free Hamiltonian acting on a line
of $n$ particles, where each $Q_i$ is a projector acting on sites
$\{i,i+1\}$.

We first restate a result by Gao, Theorem~1~1.b from
\Ref{ref:gao2015quantum}, interpreted in our context as a converse
to the DL:
\begin{lemma}[Converse of the detectability lemma]
\label{lem:LD} 
  Let $H=\sum_i Q_i$ where the $Q_i$ are projectors
  given in arbitrary order. Then for every state
  $\ket{\psi}$,
  \begin{align}
    \norm{\prod_i (\Id-Q_i)\ket{\psi}}^2 
      \ge 1- 4\bra{\psi}H\ket{\psi} .
  \end{align}  
\end{lemma}
Gao's result shows in particular that for every state
$\norm{\DL(H)\ket{\psi}}^2\ge 1-4\bra{\psi}H\ket{\psi}$. From the
discussion in the Introduction we see that this establishes that
$1-\Delta\ge \sqrt{1-4\spg}$, which implies $\Delta\le 4\spg$.
Together with the DL, it therefore shows the following relation
between the spectral gap of $H$ and that of $\DL(H)$:
\begin{align}
  \frac{\spg}{4g^2} \le \Delta \le 4\spg .
\end{align}
Moreover, up to the factors of $4$, both inequalities are tight: for
the first this is shown by the example described at the end of the
previous section, and the second is trivial.

We now turn to the definition of the coarse-grained Hamiltonian. For
this, fix an even integer $r\ge 2$ and group particles in groups of
$r$ neighboring particles. Define $S_1 = \{1,\ldots,r\}$, $S_2 =
\{r/2+1,\ldots, r/2+r\}$, and more generally $S_{\alpha} =
\{(\alpha-1) r/2 +1,\ldots,(\alpha-1) r/2+r \}$ for $\alpha\geq 1$
(see \Fig{fig:grouping} for an illustration). For each $\alpha$ let
$\oP_\alpha$ be the projector on the common ground space of all
local terms $Q_i$ that act exclusively on particles in $S_\alpha$,
and let $\oQ_\alpha\EqDef \Id-\oP_\alpha$. The coarse-grained
Hamiltonian is given by
\begin{align}
\label{def:oH}
  \bar{H} \EqDef \sum_\alpha \oQ_\alpha .
\end{align}

\begin{figure}
  \centering
  \includegraphics[scale=1]{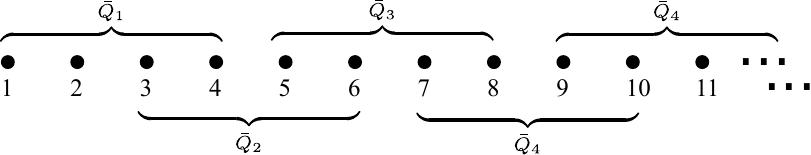}
  \caption{Grouping $r=4$ neighboring particles to subsets in order
  to define the coarse-grained Hamiltonian $\oH$. \label{fig:grouping}}
\end{figure}

Clearly, any ground state of $H$ is a ground state of $\oH$, so that
$\oH$ is frustration free. Conversely, any ground state of $\oH$ is
a ground state of $H$ as well, as for any $Q_i$ there is at least
one set $S_\alpha$ which contains both particles it acts on, so that
$\oQ_\alpha\ket{\psi} = 0 \implies Q_i\ket{\psi}=0$. The following
theorem gives a lower bound on the spectral gap of $\oH$.

\begin{theorem}
\label{thm:amp}
  The spectral gap of $\oH$ is at least $\frac{1}{4} -
  e^{-\frac{1}{2}(r-4)\sqrt{\spg/2}}$.
\end{theorem}
Before proving the theorem, we note, following
\Ref{ref:GH2015-DLAmp}, that \Thm{thm:amp} is optimal in the sense
that in general one cannot hope to amplify the gap of a
frustration-free system to a constant by coarse graining into groups
of $r=\bigO{\spg^{-\lambda}}$ particles with $\lambda<1/2$. Indeed,
as was shown in~\Ref{ref:GH2015-DLAmp}, there exists a
frustration-free 1D Hamiltonian (the XXZ model with kink boundary
conditions) for which the correlation length is
$\xi=\bOmega{\spg^{-1/2}}$. On the other hand, as shown by
Hastings\cc{ref:Hastings2004-EXP}, the correlation length of every
$r$-local Hamiltonian chain with a constant spectral gap is
$\xi=O(r)$. Hence, coarse graining the XXZ model to a length scale
$r$ that produces a constant gap necessarily requires
$r=\bOmega{\spg^{-1/2}}$.

\begin{proof}
  Let $\Vgs$ be the ground space of $H$, and let $\Vgs^\perp$ be its
  orthogonal subspace. As argued above, these are also the 
  corresponding subspaces of $\oH$. Let
  \begin{align*}
    \oPi_{\text{odd}}\EqDef(\Id-\oQ_1)\cdot(\Id-\oQ_3)\cdots 
      (\Id-\oQ_{n-1}) ,
  \end{align*}
  and
  \begin{align*}
    \oPi_{\text{even}} \EqDef(\Id-\oQ_2)\cdot(\Id-\oQ_4)
      \cdots (\Id-\oQ_n)
  \end{align*}
  be the projectors onto the ground spaces of the odd and even
  layers of $\oH$ (where we have assumed $n$ to be even), so that
  $\DL(\oH) = \oPi_{\text{even}}\cdot \oPi_{\text{odd}}$. By the converse of the
  DL (\Lem{lem:LD}), for every state $\ket{\psi}$,
  $\norm{\DL(\oH)\ket{\psi}}^2\ge 1-4\bra{\psi}\oH\ket{\psi}$.
  Consequently, for every $\ket{\psi^\perp}$ orthogonal to the
  ground space of $\oH$, we have
  $\norm{\DL(\oH)\ket{\psi^\perp}}^2\ge 1-4\ospg$, where $\ospg$ is
  the spectral gap of $\oH$. Thus
  \begin{align}
    \ospg \ge \frac{1}{4} -
      \frac{1}{4}\max_{\ket{\psi^\perp}\in\Vgs^\perp}
        \norm{{\DL(\oH)\ket{\psi^\perp}}}^2 ,\label{eq:spbound}
  \end{align}
  and to prove the theorem it will suffice to provide an upper bound
  on $\max_{\ket{\psi^\perp}\in\Vgs^\perp}
  \norm{{\DL(\oH)\ket{\psi^\perp}}}^2$. We achieve this by using the
  DL on the \emph{original} Hamiltonian $H$. To that aim, let
  $\Pi_{\text{even}}, \Pi_{\text{odd}}$ be the projectors onto the ground spaces
  of the even and odd layers in $H$ respectively. We first show the
  following:
  \begin{claim}
  \label{claim:lightcone}
    For every $0\le q\le \lfloor \frac{r}{4}\rfloor$,
    \begin{align*}
      \oPi_{\text{even}}\cdot\oPi_{\text{odd}}
        = \oPi_{\text{even}}\big(
          \Pi_{\text{even}}\Pi_{\text{odd}}\Pi_{\text{even}}\big)^q\oPi_{\text{odd}} .
    \end{align*}
  \end{claim}
  
  \begin{proof}
    For every local term in the original Hamiltonian, define
    $P_i\EqDef \Id-Q_i$ so that $\Pi_{\text{even}}$ and $\Pi_{\text{odd}}$ are
    products of $P_i$ terms. The main observation required is that
    from every coarse-grained $\oP_{\alpha}$ we can ``pull'' a
    light-cone of $P_i$ projectors either to its left or to
    its right. Suppose for instance that $\alpha=1$ and $r$ is even.
    Then by definition $P_i \oP_1 = \oP_1$ for $i=1,3,\ldots,r-1$,
    so $\oP_1 = (P_1\cdots P_{r-1})\oP_1$; more generally, 
    \begin{align*}
      \oP_1 = P_{r/2} (P_{r/2-1} P_{r/2+1}) 
        \cdots (P_2\cdots P_{r-2})(P_1\cdots P_{r-1})\oP_1 ,
    \end{align*}
    and a similar argument applies to different values of $\alpha$
    and odd $r$ as well. Pulling such light-cones from the left of
    $\oPi_{\text{odd}}$ and from the right of $\oPi_{\text{even}}$, the projectors
    can be arranged in layers to form the product
    $\big(\Pi_{\text{even}}\Pi_{\text{odd}}\Pi_{\text{even}}\big)^q$; this is
    demonstrated in \Fig{fig:coarse-grain} for $r=8$ and $q=2$.
  \end{proof}
  
  \begin{figure}[!ht]
    \centering
    \includegraphics[width=0.45\textwidth]{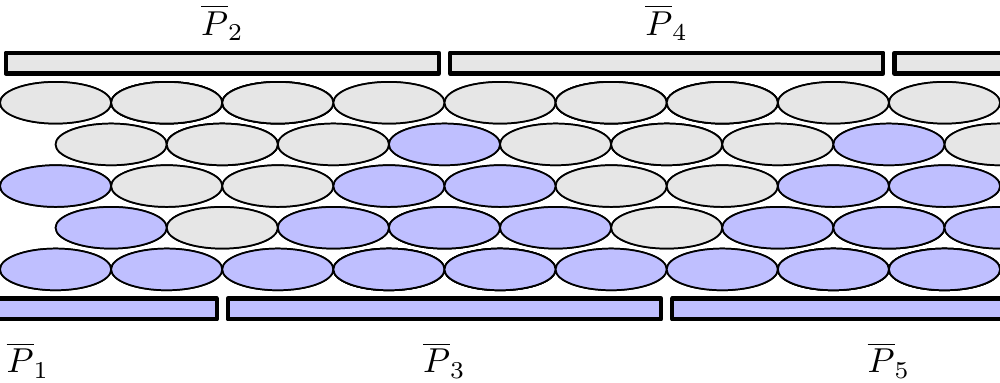}    
    \caption{Pulling out local projectors from coarse-grained
      projectors in a system with $r=8$. 
      Small ellipses are local projectors $P_i$. Wide
      rectangles are coarse-grained $\oP_\alpha$.
        \label{fig:coarse-grain} }
  \end{figure}
  
  Notice that $(\Pi_{\text{even}}\Pi_{\text{odd}}\Pi_{\text{even}}\big)^q =
  \big(\DL(H)^\dagger \DL(H)\big)^q$, so that applying the DL on $H$
  we may conclude that for any $\ket{\psi^\perp}\in \Vgs^\perp$, 
  \begin{align*}
    \bnorm{\oPi_{\text{even}}\cdot\oPi_{\text{odd}}\ket{\psi^\perp}}^2
      \le \left(\frac{1}{\spg/4+1}\right)^{2
          \lfloor\frac{r}{4}\rfloor}
      = 1-\bOmega{r\spg}.
  \end{align*}
  Together with~\eqref{eq:spbound} this is already sufficient to
  obtain a lower bound on the spectral gap of $\oH$. To improve the
  bound to the quadratic dependence on $\spg$ claimed in the theorem,
  we follow an idea from \Ref{ref:GH2015-DLAmp} of using the
  Chebyshev polynomial to boost the effect of the DL.  For the sake
  of completeness, we repeat the argument in detail.
  
  Let $A\EqDef\DL(H)^\dagger \DL(H)=\Pi_{\text{even}}\Pi_{\text{odd}}\Pi_{\text{even}}$.
  Using Claim~\ref{claim:lightcone}, for any polynomial $P_q$ of
  degree $q\le\lfloor\frac{r}{4}\rfloor$ such that $P_q(1)=1$, it
  holds that $\oPi_{\text{even}}\cdot\oPi_{\text{odd}} = \oPi_{\text{even}}\cdot
  P_q(A)\cdot \oPi_{\text{odd}}$.  By definition, for any
  $\ket{\psi^\perp}\in \Vgs^\perp$, we have
  $\oPi_{\text{odd}}\ket{\psi^\perp}\in\Vgs^\perp$ (to see this, multiply
  from the left by any ground state of $\oH$). Using that
  $\norm{\oPi_{\text{even}}}\le 1$, we conclude
  \begin{align}
  \nonumber
    &\max_{\ket{\psi^\perp}} 
      \bnorm{\oPi_{\text{even}}\cdot \oPi_{\text{odd}}\ket{\psi^\perp}} \\
   &\ \ \le \max_{\ket{\psi^\perp}} 
     \bnorm{\oPi_{\text{even}}\cdot
        P_q(A)\cdot\oPi_{\text{odd}}\ket{\psi^\perp}} \nonumber\\
    &\ \ \le \max_{\ket{\psi^\perp}} 
      \bnorm{P_q(A)\ket{\psi^\perp}} .
  \label{eq:spbound2}
  \end{align}
  Our goal is therefore to find a polynomial $P_q(x)$ that would
  minimize the RHS of the above inequality.  Since $A$ is Hermitian,
  we may expand $\ket{\psi^\perp}$ in a basis of eigenstates of $A$
  as $\ket{\psi^\perp}=\sum_\mu \psi_\mu \ket{\mu}$. By definition,
  $0 \le A \le \Id$, and so its eigenvalues are in the range
  $[0,1]$. The $\mu=1$ eigenvalue corresponds to the ground space of
  $H$, and since $A= \DL(H)^\dagger \DL(H)$, it follows from the DL
  that all other eigenvalues of $A$ are upper bounded by
  $h\EqDef\frac{1}{\spg/4+1}$. We therefore look for a polynomial
  $P_q(x)$ with $q\le\lfloor\frac{r}{4}\rfloor$ such that $P_q(1)=1$
  and $|P(x)|$ is minimal for $x\in [0,h]$. Following the approach
  of the AGSP-based area-law proofs\cc{ref:ALV2012-AL,
  ref:AKLV2013-AL}, we choose $P_q$ to be a rescaled Chebyshev
  polynomial of degree $q$ of the first kind. The exact construction
  is summarized in Lemma~\ref{lem:chebyshev} given at the end of
  this section. Substituting $h=\frac{1}{\spg/4+1}$ in the lemma and
  noticing that as $\spg\le g+1=3$ (see the discussion following
  \Cor{cor:DL} for a justification), it follows that
  $1-h=1-\frac{1}{\spg/4+1}\ge \spg/8$, and consequently for every
  $x\in[0,h]$, 
  \begin{align}
    |P_q(x)| \le 2e^{-q\sqrt{\spg/2}} .
  \end{align}
  Therefore, $\norm{P_q(A)\ket{\psi^\perp}}\le 2e^{-q\sqrt{\spg/2}}$
  for every $\ket{\psi^\perp}\in\Vgs^\perp$, and
  combining~\eqref{eq:spbound} and~\eqref{eq:spbound2},
  \begin{align*}
    \ospg \ge \frac{1}{4} - e^{-2q\sqrt{\spg/2}} .
  \end{align*}
  Finally, the theorem is proved by choosing
  $q=\lfloor\frac{r}{4}\rfloor\ge \frac{r}{4}-1$.
\end{proof}

\begin{lemma}
\label{lem:chebyshev} 
  Let $0<h<1$, and let $T_q(x)$ be the Chebyshev polynomial of the
  first kind of degree $q$. Define 
  \begin{align*}
    \tilde{P}_q(x) &\EqDef T_q(2\frac{x}{h}-1) , &
    P_q(x)&\EqDef \tilde{P}_q(x)/\tilde{P}_q(1) .
  \end{align*}  
  Then $P_q(1)=1$ and for any $x\in[0,h]$ it holds that
  $|P_q(x)|\le 2e^{-2q\sqrt{1-h}}$.
\end{lemma}

\begin{proof}
  $P_q(1)=1$ holds by definition.  Using the well-known properties
  of the Chebyshev polynomial (see, for example, Lemma~4.1 in
  \Ref{ref:AKLV2013-AL}),
  \begin{align*}
    |T_q(x)|&\le 1, & &\text{for $|x|\le 1$} ,\\
    |T_q(x)| &\ge \frac{1}{2}\exp{
      \left(2q\sqrt{(|x|-1)/(|x|+1)}\right)} ,
      & &\text{for $|x|>1$}  ,
  \end{align*}
  it is easy to see that $|\tilde{P}_q(1)| \ge
  \frac{1}{2}e^{2q\sqrt{1-h}}$, and therefore for $x\in[0,h]$ we
  have $|P_q(x)|\le 2e^{-2q\sqrt{1-h}}$.
\end{proof}

%%%%%%%%%%%%%%%%%%%%%%%%%%%%%%%%%%%%%%%%%%%%%%%%%%%%%%%%%%%%%%%%
\section{Summary}

We have provided a short proof of the DL which tightens its bound
and generalizes it to arbitrary orderings of the local projectors.
Using an explicit example, we showed that the new bound is optimal
in its dependence on $g$ when $\epsilon_\phi\to 0$, up to constant
factors. In addition, we have shown how the lemma can be combined
with a converse bound to prove that by coarse graining a
frustration-free Hamiltonian with a gap $\spg>0$ to a length scale
$\bigO{\spg^{-1/2}}$, one obtains a Hamiltonian with a constant
spectral gap. It would be interesting to see if, by using the
converse to the DL, one can apply the DL to slightly frustrated
systems with a constant gap in a controlled manner. If this can be
done, it would extend the applicability DL to much broader set of
problems, which may benefit from its simplicity with respect to
other techniques.  

\begin{acknowledgments}
  We thank Zeph Landau for many insightful discussions, and Mark
  Wilde for bringing \Ref{ref:gao2015quantum} to our attention. We
  also thank an anonymous referee for pointing out minor
  imprecisions in an earlier draft of this paper. T.V. was partially
  supported by the IQIM, an NSF Physics Frontiers Center (NFS Grant
  No.~PHY-1125565) with support of the Gordon and Betty Moore
  Foundation (GBMF-12500028). A.A. was supported by Core grants of
  Centre for Quantum Technologies, Singapore. Research at the Centre
  for Quantum Technologies is funded by the Singapore Ministry of
  Education and the National Research Foundation Singapore, also
  through the Tier 3 Grant random numbers from quantum processes.
\end{acknowledgments}

\bibliographystyle{apsrev4-1}
\bibliography{DLD}

\end{document}